\documentclass[a4paper,11pt,single]{llncs}
\usepackage{time,amsmath,marvosym,algorithm,algorithmic,subfigure,subfig,graphicx,times,url,comment,setspace}
\newtheorem{ourClaim}{Claim}{\bfseries}{ }

\begin{document}

\title{Computing a Longest Common Palindromic Subsequence}

\author{Shihabur Rahman Chowdhury \and Md. Mahbubul Hasan \and Sumaiya Iqbal \and M. Sohel Rahman}

\institute{A$\ell$EDA group\\Department of CSE, BUET, Dhaka - 1000, Bangladesh\\
\email{\{shihab,mahbubul,sumaiya,msrahman\}@cse.buet.ac.bd}}

\maketitle

\begin{abstract}
The {\em longest common subsequence (LCS)} problem is a classic and well-studied problem in computer science. 
Palindrome is a word which reads the same forward as it does backward. The {\em longest common palindromic subsequence (LCPS)} problem is an interesting variant of the classic LCS problem which finds the longest common subsequence between two given strings such that the computed subsequence is also a palindrome. In this paper, we study the LCPS problem and give efficient algorithms to solve this problem. To the best of our knowledge, this is the first attempt to study and solve this interesting problem.\\

\end{abstract}

\section{Introduction}
\label{intro}

The {\em longest common subsequence (LCS)} problem is a classic and well-studied problem in computer science with a lot of variants arising out of different practical scanarios. In this paper, we introduce and study the {\em longest common palindromic subsequence (LCPS)} problem. A {\em subsequence} of a string is obtained by deleting zero or more symbols of that string. A {\em common subsequence} of two strings is a subsequence common to both the strings. A \emph{palindrome} is a word, phrase, number, or other sequence of units which reads the same forward as it does backward. The LCS problem for two strings is to find a common subsequence in both the strings, having maximum possible length. In the LCPS problem, the computed longest common subsequence, i.e., LCS, must also be a palindrome. More formally, given a pair of strings $X$ and $Y$ over the alphabet $\Sigma$, the goal of the LCPS problem is to compute a Longest Common Subsequence $Z$ such that $Z$ is a palindrome.
%
%
%
In what follows, for the sake of convenience we will assume, that $X$ and $Y$ have equal length, $n$. But our result can be easily extended to handle two strings of different length.

String and sequence algorithms related to palindromes have attracted stringology researchers since long~\cite{Apo95,Bre95,h,Galil76,p,o,k,Mart09,i,n,m}. The LCPS problem also seems to be a new interesting addition to the already rich list of problems related to palindromes. Apart from being interesting from pure theoretical point of view, LCPS has motivation from computational biology as well. Biologists believe that palindromes play an important role in regulation of gene activity and other cell processes because these are often observed near promoters, introns and specific untranslated regions. So, finding common palindromes in two genome sequences can be an important criterion for comparing them, and also to find common relationships between them.

To the best of our knowledge, there exists no research work in the literature on computing longest common palindromic subsequences. However, the problem of computing palindromes and variants in a single sequence has received much attention in the literature. Martnek and Lexa studied faster palindrome searching methods by hardware acceleration~\cite{i}. They showed that their results are better than software methods. A searching method for palindromic sequences in the primary structure of protein was presented in~\cite{j}. Manacher discovered an on-line sequential algorithm that finds all \textit{`initial'\footnote{A string $X[1 \ldots n]$ is said to have an initial palindrome of length $k$ if the prefix $S[1 \ldots k]$ is a palindrome.}} palindromes in a string~\cite{k}. Gusfield gave a linear-time algorithm to find all \textit{`maximal' palindromes} in a string~\cite{l}. Porto and Barbosa gave an algorithm to find all \textit{approximate palindromes} in a string~\cite{m}. In~\cite{h}, a simple web based tool is presented to assist biologist to detect palindromes in a DNA sequence. Authors in~\cite{n} solved the problem of finding all palindromes in SLP (Straight Line Programs)-compressed strings. Additionally, a number of variants of palindromes have also been investigated in the literature~\cite{p,t,o}. Very recently, Tomohiro \emph{et. al.} worked on pattern matching problems involving palindromes~\cite{u}.
%

\subsection{Our Contribution}

In this paper, we introduce and study the LCPS problem. We, propose two methods for finding an LCPS, given two strings. Firstly we present a dynamic programming algorithm to solve the problem with time complexity $O(n^4)$, where $n$ is the size of the strings. Then, we present another algorithm that runs in $\mathcal{O}(\mathcal{R}^2\log^3n)$ time. Here, the set of all ordered pair of matches between two strings is denoted by $\mathcal{M}$ and $|\mathcal{M}| = \mathcal{R}$.

\section{Preliminaries}
\label{prl}
%
%
We assume a finite alphabet, $\Sigma$. For a string $X = x_1 x_2 \ldots x_n$, we define the string $x_i \ldots x_j$ ($1 \leq i \leq j \leq n$) as a \emph{substring} of $X$ and denote it by $X_{i,j}$. A \emph{palindrome} is a string which reads the same forward and backward. We say a string $Z = z_1 z_2 \ldots z_u$ is a palindrome \emph{iff} $z_i = z_{u - i + 1}$ for any $ 1 \leq i \leq \left\lceil \frac{u}{2} \right\rceil $. A \emph{subsequence} of a string $X$ is a sequence obtained by deleting zero or more characters from $X$. A subsequence $Z$ of $X$ is a \emph{palindromic subsequence} if $Z$ is a palindrome.
For two strings $X$ and $Y$, if a common subsequence $Z$ of $X$ and $Y$ is a palindrome, then $Z$ is said to be a \emph{common palindromic subsequence (CPS)}. A \emph{CPS} of two strings having the maximum length is called the \emph{Longest Common Palindromic Subsequence (LCPS)} and we denote it by $LCPS(X,Y)$.
%
%

For two strings $X = x_1 x_2 \ldots x_n$ and $Y = y_1 y_2 \ldots y_n$ we define a \emph{match} to be an ordered pair $(i,j)$ such that $X$ and $Y$ has a matching character at that position respectively, that is $x_i$ = $y_j$. The set of all matches between two strings $X$ and $Y$ is denoted by $\mathcal{M}$ and it is defined as, $\mathcal{M} = \{ (i,j) : 1 \leq i \leq n , 1 \leq j \leq n \text{ and } x_i = y_j \}$. And we have, $|\mathcal{M}| = \mathcal{R}$. We define, \emph{$\mathcal{M}_\sigma$} as a subset of $\mathcal{M}$ such that all matches within this set match to a single character $\sigma \in \Sigma$. That is, $\mathcal{M}_\sigma = \{ (i,j) : 1 \leq i \leq n , 1 \leq j \leq n \text{ and } x_i = y_j = \sigma \in \Sigma\}$. And we also have $|\mathcal{M}_\sigma|=\mathcal{R}_\sigma$. Clearly, $\mathcal{M}_\sigma \subset \mathcal{M}$ and $\mathcal{M} = \displaystyle\bigcup_{\sigma \in \Sigma} \mathcal{M}_\sigma$. Each member of $\mathcal{M}_\sigma$ is called a \emph{$\sigma$-match}.

\section{A Dynamic Programming Algorithm}
\label{dpAlgo}
A brute-force approach to this problem would be to enumerate all the subsequences of $X$ and $Y$ and compare them, keeping track of the longest palindromic subsequence found. There are $2^n$ subsequences of any string of length $n$. So the brute force approach would lead to an exponential time algorithm. In this section, we will devise a dynamic programming algorithm for the LCPS problem. Here, we will see that the natural classes of subproblems for LCPS correspond to pairs of \emph{substrings} of the two input sequences. We first present the following theorem which proves the optimal substructure property of the LCPS problem.

\begin{theorem}
\label{thDp}
Let $X$ and $Y$ are two sequences of length $n$ and $X_{i,j} = x_i x_{i+1} \ldots x_{j-1} x_j$ and $Y_{k,\ell} = y_k y_{k+1} \ldots y_{\ell-1} y_\ell$ are two substrings of those respectively. Let $Z = z_1 z_2 \ldots z_u$ be the \emph{LCPS} of the two substrings, $X_{i,j}$ and $Y_{k,\ell}$. Then, the following statements hold,
\begin{enumerate}
	\item If $x_i = x_j = y_k = y_\ell = a$ ($a \in \Sigma$), then $z_1 = z_u = a$ and $z_2 \ldots z_{u - 1}$ is an \emph{LCPS} of $X_{i + 1,j - 1}$ and $Y_{k + 1,\ell - 1}$.
	\item If $x_i = x_j = y_k = y_l$ condition does not hold then, $Z$ is an \emph{LCPS} of  ($X_{i + 1,j}$ and $Y_{k,\ell}$) or ($X_{i,j - 1}$ and $Y_{k,\ell}$) or ($X_{i,j}$ and $Y_{k,\ell - 1}$) or ($X_{i,j}$ and $Y_{k + 1,\ell}$).
\end{enumerate}
\end{theorem}

\begin{proof}
\textbf{(1)} Since $Z$ is a palindrome by definition so we have $z_1 = z_u$. If $z_1 = z_u \neq a$ then we can append $a$ at both ends of $Z$ to obtain a common palindromic subsequence of $X_{i,j}$ and $Y_{k,\ell}$ of length $u + 2$, which contradicts the assumption that $Z$ is the \emph{LCPS} of $X_{i,j}$ and $Y_{k,\ell}$. So we must have $z_1 = z_u = a$. Now, the substring $z_2 \ldots z_{u - 1}$ with length $u - 2$ is itself a palindrome and it is common to both $X_{i + 1,j - 1}$ and $Y_{k + 1,\ell - 1}$. We need to show that it is an \emph{LCPS}. For the purpose of contradiction let us assume that there is an common palindromic subsequence $W$ of $X_{i + 1,j - 1}$ and $Y_{k + 1,\ell - 1}$ with length greater than $u - 2$. Then appending $a$ to both ends of $W$ will produce a common subsequence of $X_{i,j}$ and $Y_{k,\ell}$ with length greater than $u$, which is a contradiction.

\textbf{(2)} Since $Z$ is a palindrome so $z_1 = z_u$. Here we have that the condition $x_i = x_j = y_k = y_\ell$ does not hold. So $z_1$ or $z_2$ is not equal to at least one of $x_i$ or $x_j$ or $y_k$ or $y_\ell$. Therefore $Z$ is a common palindromic subsequence of the substrings obtained by deleting at least one character from either end of $X_{i,j}$ or $Y_{k,\ell}$. If any pair of substrings obtained by deleting one character from either end of $X_{i,j}$ or $Y_{k,\ell}$ has a common palindromic subsequence $W$ with length greater than $u$ then it would also be a common palindromic sub-sequence of $X_{i,j}$ and $Y_{k,\ell}$, contradicting the assumption that $Z$ is a \emph{LCPS} of $X_{i,j}$ and $Y_{k,\ell}$.

This completes the proof.\qed
\end{proof}
%
From Theorem~\ref{thDp}, we see that if $x_i = x_j = y_k = y_\ell = a$ ( $a \in \Sigma$ ), we must find an \emph{LCPS} of $X_{i + 1,j - 1}$ and $Y_{k + 1,\ell - 1}$ and append $a$ on its both ends to yield the \emph{LCPS} of $X_{i,j}$ and $Y_{k,\ell}$. Otherwise, we must solve four subproblems and take the maximum of those. These four subproblems correspond to finding \emph{LCPS} of:

(a) $X_{i + 1,j}$ and $Y_{k,\ell}$ (b) $X_{i,j - 1}$ and $Y_{k,\ell}$ (c) $X_{i,j}$ and $Y_{k,\ell - 1}$ and (d) $X_{i,j}$ and $Y_{k + 1,\ell}$

Let us define $lcps[i,j,k,\ell]$ to be the length of the \emph{LCPS} of  $X_{i,j}$ and $Y_{k,\ell}$. If either $i > j$ or $k > \ell$ then one of the substrings is empty and hence the length of our \emph{LCPS} is $0$. So we have,

\begin{align}
\label{dp3}
lcps[i,j,k,\ell] = 0 \hspace{0.3in} \text{if $i > j$ or $k > \ell$}
\end{align}

If either of the substrings has length $1$, then the obtained \emph{LCPS} will have length $0$ or $1$ depending on whether
that single character can form a common palindrome between the two substrings. In this case we have,

\begin{align}
\label{dp4}
lcps[i,j,k,\ell] = 1& &\text{if ($i = j$ or $k = \ell$) and (either of $x_i$ or $x_j$ equals either of $y_k$ or $y_\ell$)}
\end{align}

Using the base cases of Equations~\ref{dp3} and \ref{dp4} and the optimal substructure property of LCPS (Theorem~\ref{thDp}), we have the following recursive formula:

\begin{equation}
\label{dp}
\displaystyle
lcps[i,j,k,\ell] =
\begin{cases}
0 & \text{$i > j$ or $k > \ell$}\\
1 & \text{($i = j$ or $k = \ell$)}\\
  & \text{ and} \\
  & \text{(either of $x_i$ or $x_j$}\\
  & \text{equals}\\
  &\text{either of $y_k$ or $y_\ell$) }\\
2 + lcps[i + 1,j - 1,k + 1,\ell - 1] & \text{($i < j$ and $k < \ell$)}\\
	&\text{and}\\
	& \text{$x_i = x_j = y_k = y_\ell$}\\
\max( lcps[i + 1,j,k,\ell], lcps[i,j - 1,k,\ell], & \\
 lcps[i,j,k + 1,\ell], lcps[i,j,k,\ell - 1]) & \text{($i < j$ and $k < \ell$)}\\
	& \text{and}\\
	& \text{the condition ($x_i = x_j = y_k = y_\ell$)}\\
	& \text{does not hold}
\end{cases}
\end{equation}

The length of an \emph{LCPS} between $X$ and $Y$ shall be stored at $lcps[1,n,1,n]$.
%
\begin{algorithm}[t!]
\small
\caption{LCPSLength(X,Y)}
\label{alg1}
\begin{algorithmic}[1]
\STATE $n \leftarrow length[X]$

\FOR{ $i = 1$ to $n$}
	\FOR{ $j = 1$ to $i$ }
		\FOR{ $k = 1$ to $n$ }
			\FOR{ $\ell = 1$ to $k$ }
				\IF{ ($i = j$ or $k = \ell$) and (either of $x_i$ or $x_j$ equals either of $y_k$ or $y_\ell$)) }
					\STATE $lcps[i,j,k,\ell] = 1$
				\ELSE
					\STATE $lcps[i,j,k,\ell] = 0$
				\ENDIF
			\ENDFOR
		\ENDFOR
	\ENDFOR
\ENDFOR

\FOR{ $xLength = 2$ to $n$ }
	\FOR{ $yLength = 2$ to $n$ }
		\FOR{ $i = 1$ to $n - xLength + 1$ }
			\FOR{ $k = 1$ to $n - yLength + 1$ }
				\STATE $j = i + xLength$
				\STATE $\ell = k + yLength$
				\IF{ $x_i = x_j = y_k = y_\ell$ }
					\STATE $lcps[i,j,k,\ell] = 2 + lcps[i + 1,j - 1,k + 1,\ell - 1]$
				\ELSE
					\STATE $lcps[i,j,k,\ell] = \max(lcps[i+1,j,k,\ell],lcps[i,j-1,k,\ell],lcps[i,j,k+1,\ell],lcps[i,j,k,\ell-1])$
				\ENDIF
			\ENDFOR
		\ENDFOR
	\ENDFOR
\ENDFOR

\RETURN $lcps$
\end{algorithmic}
\end{algorithm}
Since there are $\Theta(n^4)$ distinct subproblems, we can use dynamic programming to compute the solution in a bottom up manner. Algorithm~\ref{alg1} outlines the \emph{LCPSLength} procedure which takes two sequences $X$ and $Y$ as inputs. It stores the $lcps[i,j,k,\ell]$ values in the $n \times n \times n \times n$ size table $lcps$. The table entries $i > j$ , $k > \ell$ has value $0$ since these entries correspond to at least one empty substring. We proceed in our computation with increasing length of the substrings. That is, table entries for substrings of length $v$ are already computed before substrings of length $v + 1$. The procedure returns the $lcps$ table and $lcps[1,n,1,n]$ contains the length of an \emph{LCPS} of $X$ and $Y$. Theorem~\ref{thDpRun} gives us the running time of Algorithm~\ref{alg1}.

\begin{theorem}
\label{thDpRun}
$LCPSLength(X,Y)$ computes the length of an LCPS of $X$ and $Y$ in $\mathcal{O}(n^4)$ time.
\end{theorem}

\begin{proof}
The initialization step takes $\mathcal{O}(n^4)$ time. As the algorithm proceeds, it computes the \emph{LCPS} of substrings of $X$ and $Y$ in such a way that substrings of length $v$ is considered before substrings of length $v + 1$. Now, there are $\mathcal{O}(n^2)$ possible pairs of lengths between $X$ and $Y$. For each of these pairs there are $\mathcal{O}(n^2)$ possible start position pairs. So the four nested loops in Lines 15 - 18 requires $\mathcal{O}(n^4)$ time. And each table entry takes $\mathcal{O}(1)$ time to compute. So the table computation takes $\mathcal{O}(n^4)$ time. \qed
%
\end{proof}
%
%
We can use the lengths computed in $lcps$ table returned by \emph{LCS-Length} to construct an \emph{LCPS} of $X$ and $Y$. We simply begin at $lcps[1,n,1,n]$ and trace back through the table. As soon as we find that $x_i = x_j = y_k = y_\ell$, we find an element of LCPS, and recursively try to find the \emph{LCPS} for $X_{i + 1, j - 1}$ and $Y_{k + 1, \ell - 1}$. Otherwise, we find the maximum value in the $lcps$ table for $(X_{i + 1,j},Y_{k,\ell})$, $(X_{i,j - 1},Y_{k,\ell})$, $(X_{i,j},Y_{k + 1,\ell})$, $(X_{i,j},Y_{k,\ell - 1})$ and then use that value to compute subsequent members of LCPS recursively. Since at least one of $i,j,k,\ell$ is decremented in each recursive call, this procedure takes $\mathcal{O}(n)$ time to construct an \emph{LCPS} of $X$ and $Y$.

\section{A Second Approach}
\label{rAlgo}

In this section, we present a second approach to efficiently solve the LCPS problem. In particular, we will first reduce our problem to a geometry problem and then solve it with the help of a balanced binary search tree data structure. The resulting algorithm will run in $O(\mathcal R^2 \log^3 n)$ time. Recall that, $\mathcal R$ is the number of ordered pairs at which the two strings match. 
%
First we make the following claim.

\begin{ourClaim}
\label{claimMatch}
Any common palindromic subsequence $Z = z_1 z_2 \ldots z_u$ of two strings $X$ and $Y$ can be decomposed into a set of $\sigma$-match pairs ($\sigma \in \Sigma$).
\end{ourClaim}

\begin{proof}
Since $Z$ is a palindrome itself so we have, $z_i = z_{u - i + 1}$ for $ 1 \leq i \leq \left\lceil \frac{u}{2} \right\rceil$. Since $Z$ is common to both $X$ and $Y$, each $z_i$ , $1 \leq i \leq u$ corresponds to a $\sigma$-match between $X$ and $Y$.
Therefore, $z_i$ and $z_{u - i + 1}$ is a $\sigma$-match pair. Now we can obtain $\sigma$-match pairs by pairing up each $z_i$ and $z_{u - i + 1}$ for all $ 1 \leq i \leq \left\lceil \frac{u}{2} \right\rceil$. So we have decomposed $Z$ into a set of $\sigma$-match pairs. \qed
\end{proof}

It follows from Claim~\ref{claimMatch} that constructing a common palindromic subsequence of two strings can be seen as constructing an appropriate set of $\sigma$-match pairs between the input strings. An arbitrary pair of $\sigma$-match, $\langle (i,j),(k,\ell) \rangle$ (say $m_1$), from among all pair of $\sigma$-matches between a pair of strings, can be seen as inducing a substring pair in the input strings. Now suppose we want to construct a common palindromic subsequence $Z$ with length $u$ with $m_1$ at the two ends of $Z$. Clearly we have $z_1 = z_u = x_i = x_j = y_k = y_\ell$. Then to compute $Z$, we will have to recursively select $\sigma$-match pairs between the induced substrings $X_{i,j}$ and $Y_{k,\ell}$. In this way we shall get a set of $\sigma$-matches which will correspond to the common palindromic sub-sequences of the input strings. If we consider all possible $\sigma$-match pairs as the two end points of the common palindromic sub-sequence then the longest obtained one among all these will be an \emph{LCPS} of the input strings. This is the basic idea for constructing \emph{LCPS} in our new approach.

To compute $\mathcal M_\sigma$ for any $\sigma \in \Sigma$, we first linearly scan $X$ and $Y$ to compute two arrays, $X_\sigma$ and $Y_\sigma$, which contains the indices in $X$ and $Y$ where $\sigma$ occurs. Then we take each pair between the two arrays to get all the ordered pairs where $\sigma$ occurs in both strings.

\subsection{Mapping the LCPS Problem to a Geometry Problem}
\label{rAlgoGeoMap}
Each match between the strings $X$ and $Y$ can be visualized as a point on a $n \times n$ rectangular grid where all the co-ordinates have integer values. Then, any rectangle in the grid corresponds to a pair of substrings of $X$ and $Y$. Any $\sigma$-match pair defines two corner points of a rectangle and thus induces a rectangle in the grid. Now, our goal is to take a pair of $\sigma$-matches as the two ends of common palindromic sub-sequence and recursively construct the set of pair of $\sigma$-matches from within the induced substrings. Clearly, the rectangle induced by a pair of $\sigma$-matches will in turn contain some points (i.e matches) as well. We recursively continue within the induced sub-rectangles to find the \emph{LCPS} between the substrings induced by the rectangles. When the recursion unfolds, we append the $\sigma$-match pair on the obtained sequence to get the \emph{LCPS} that can be obtained with our $\sigma$-match pair corresponding to the two ends. Clearly, if we do this procedure for all such possible $\sigma$-match pairs then the longest of them will be our desired \emph{LCPS} between the two strings. The terminating condition of this recursive procedure would be:

\begin{enumerate}
	\item[T1. ] If there is no point within any rectangle. This corresponds to the case where at least one of the substrings is empty.
	\item[T2. ] If it is not possible to take any pair of $\sigma$-matches within any rectangle. This corresponds to the single character case in our Dynamic Programming solution.
\end{enumerate}

So, in summary we do the following.
\begin{enumerate}
\item Identify an induced rectangle (say $\Psi_1$) by a pair of $\sigma$-matches.
\item Pair up $\sigma$-matches within $\Psi_1$ to obtain another rectangle (say $\Psi_2$) and so on until we encounter either of the two terminating conditions T1 or T2.
\item We repeat the above for all possible $\sigma$-match pairs ($\forall \sigma \in \Sigma$).
\item At this point, we have a set of nested rectangle structures.
\item Here, an increase in the nesting depth of the rectangle structures as it is being constructed, corresponds to adding a pair of symbols\footnote{If condition T2 is reached, only a symbol shall be added.} to the resultant palindromic subsequence. Hence, the set of rectangles with maximum nesting depth gives us an \emph{LCPS}.
\end{enumerate}

Now our problem reduces to the following interesting geometric problem: \emph{Given a set of nested rectangles defined by the $\sigma$-match pairs $\forall \sigma \in \Sigma$, we need to find the set of rectangles having the maximum nesting depth. }

In what follows, we will refer to this problem as the Maximum Depth Nesting Rectangle Structures (MDNRS) problem.

\subsection{A Solution to the MDNRS Problem}
\label{rAlgoSegSol}
A $\sigma$-match pair, $\langle (i,j), (k,\ell) \rangle$ basically represents a $2$-dimensional rectangle (say $\Psi$). Assume, without the loss of generality that $(i,j)$ and $(k,\ell)$ correspond to the lower left corner and upper right corner of $\Psi$, respectively. In what follows, depending on the context, we will sometimes use $\langle (i,j), (k,\ell) \rangle$ to denote the corresponding rectangle. Now, a rectangle $\Psi' (\langle (i',j'),(k',l')\rangle)$ will be nested within rectangle $\Psi (\langle (i,j),(k,l) \rangle)$ \emph{iff} the following condition holds:

$i' > i$ and $j' > j$ and $k' < k$ and $\ell' < \ell$ \\
$\Leftrightarrow$ $i' > i$ and $j' > j$ and $-k' > -k$ and $-\ell' > -\ell$ \\
$\Leftrightarrow$ $(i',j',-k',-\ell') > (i,j,k,\ell)$.

Now we convert a $2$-dimensional rectangle $\Psi (\langle (i,j),(k,\ell) \rangle)$ to a $4$-dimensional point $P_{\Psi}(i,j,-k,-\ell)$. We say that a point $(x,y,z,w)$ is chained to another point $(x',y',z',w')$ \emph{iff} $(x,y,z,w) > (x',y',z',w')$. Then, it is easy to see that, a rectangle $\Psi' (\langle (i',j'),(k',\ell')\rangle)$, is nested within a rectangle $\Psi (\langle (i,j),(k,\ell)\rangle)$ \emph{iff} the point $P_{\Psi'}(i',j',-k',-\ell')$ is chained to the point $P_{\Psi}(i,j,-k,-\ell)$. Hence, the problem of finding the set of rectangles in $2$-dimension having the maximum nesting depth easily reduces to finding the set of corresponding points in $4$-dimension having the maximum chain length.

First we give a solution of this problem for $2$-dimension. Later we shall extend our solution for $4$-dimension. In $2D$ our points will be in the form of $(x,y)$. We maintain a $1$-dimensional balanced binary search tree $\mathcal{T}$ that will contain the $x$ coordinate of the points along with a value as the points are being processed. The value indicates the length of longest chain that can be formed starting from any point with that $x$ co-ordinate. Initially $\mathcal{T}$ is empty. We process the points in non-increasing order of their $y$ coordinates. For each point $(x,y)$ we make a query to $\mathcal{T}$ for the $x'$ such that $x'$ is the smallest number that is greater than $x$ (i.e., a \emph{successor query}). If the value corresponding to $x'$ is $K$ then we can construct a chain of length $K + 1$ starting from the point $(x,y)$, and which will immediately preceded a point with $x'$ as its $x$-coordinate. Now we insert/update $x$ in the tree with corresponding value $K + 1$. Since $\mathcal{T}$ is balanced, any insertion, deletion and successor query operation can be done in $\mathcal{O}(\log{n})$ time. The maximum value in $\mathcal{T}$ is the maximum length of the chain which in turn will yield the length of \emph{LCPS} between the input sequences. If we also store at $x$ the point $(x,y)$, which yields the maximum chain length then we can use that to trace the chain later in linear time to get the sequence as well.

We can extend our $1$-Dimensional balanced binary tree to $d$-dimension in the form of multi-level trees using an inductive definition on $d$. In $d$-dimension we shall store $(x_1,x_2, \ldots x_{d-1})$ in $\mathcal{T}$ with respect to $x_d$-coordinates. For all nodes $u$ of $\mathcal{T}$, we associate a $(d - 1)$-dimensional multi-level balanced binary search tree with respect to $(x_1,x_2, \ldots x_{d - 1})$. During insertion, deletion and search operations for $d$-dimensional points we also perform the same operation recursively in the $d - 1$-dimensional trees. By induction on $d$ it can be trivially shown that the insertion, deletion and searching in this balanced multi-level binary search tree can be done in $\mathcal{O}(\log^d{n})$ time.

Finally to solve our problem we simply use a $3$-dimensional balanced multi-level binary search tree. Now we process the points $(x,y,z,w)$ in non-increasing order of the highest dimension $w$. For each point $(x,y,z,w)$ we query the tree for $(x',y',z')$ such that $x > x'$, $y > y'$ , $z > z'$ and $x', y', z'$ are the smallest number greater than $x,y,$ and $z$ respectively. The rest of the process are same.

\begin{algorithm}[t!]
\small
\caption{LCPS-New(X,Y,$\Sigma$)}
\label{alg3}
\begin{algorithmic}[1]

\FOR{each $\sigma \in \Sigma$}
	\STATE $\mathcal{M}_\sigma \leftarrow \phi$
	\STATE $X\sigma \leftarrow \phi$
	\STATE $Y\sigma \leftarrow \phi$
	\FOR{$i = 1$ to $n$}
		\IF{$X[i] = \sigma$}
			\STATE $X\sigma \leftarrow X\sigma \cup \{i\}$
		\ELSIF{$Y[i] = \sigma$}
			\STATE $Y\sigma \leftarrow Y\sigma \cup \{i\}$
		\ENDIF
	\ENDFOR
	
	\FOR{$i = 1$ to $|X\sigma|$}
		\FOR{$j = 1$ to $|Y\sigma|$}
			\STATE $\mathcal{M}_\sigma \leftarrow \mathcal{M}_\sigma \cup \{(X\sigma[i],Y\sigma[i])\}$
		\ENDFOR
	\ENDFOR
\ENDFOR

\STATE $Rectangles \leftarrow \phi$ \COMMENT{$Rectangles$ contains the set of all rectangles}

\FOR{each $\sigma \in \Sigma$}
	\FOR{each match $(i,j) \in \mathcal{M}_\sigma$}
		\FOR{each match $(k,\ell) \in \mathcal{M}_\sigma$}	
			\STATE $Rectangles \leftarrow Rectangles \cup \{(i,j),(k,\ell)\}$
		\ENDFOR
	\ENDFOR
\ENDFOR

\STATE $P \leftarrow \phi$
\FOR{each $\Psi(i,j,k,\ell) \in Rectangles$}
	\STATE $\mathcal{P} \leftarrow \mathcal{P} \cup \{(i,j,-k,-l)\}$
\ENDFOR
\STATE Sort the points in $\mathcal{P}$ in non increasing order of 4th dimension
\STATE Initialize the multi-level balanced binary search tree $\mathcal{T}$ as empty tree

\FOR{each point $p(i,j,k,l) \in \mathcal{P}$}
	\STATE Find the point $(i',j',k')$ in $\mathcal{T}$ such that $i' > i$ and $j' > j$ and $k' > k$ and $i', j', k'$ are the smallest integer greater than $i, j,$ and $k$ respectively.
	\STATE $K \leftarrow$ the value stored at $(i',j',k')$
	\IF{$(i,j,k)$ exists in $\mathcal{T}$}
		\STATE Update the value of $(i,j,k)$ with $K + 1$
	\ELSE
		\STATE Insert the node $(i,j,k)$ with value $K + 1$
	\ENDIF

	\STATE Also store $(i',j',k')$ in $\mathcal{T}$ at the node $(i,j,k)$ as its successor.
\ENDFOR
\STATE $lcps \leftarrow$ maximum value stored in $\mathcal{T}$
\STATE $LCPS \leftarrow$ trace the successors to obtain the sequence
\RETURN $LCPS$
\end{algorithmic}
\end{algorithm} 
Algorithm~\ref{alg3} outlines the \emph{LCPS-New} procedure which takes as input two strings $X$ and $Y$, each of length $n$ and the alphabet, $\Sigma$.
The following theorem gives the worst case running time of \emph{LCPS-New} procedure.

\begin{theorem}
The LCPS-New procedure computes an LCPS of strings $X$ and $Y$ in $\mathcal{O}(\mathcal{R}^2 \log^3{n})$ time.
\end{theorem}

\begin{proof}
Since there are $\mathcal{R}$ matches between $X$ and $Y$, we have $\mathcal{O}(\mathcal{R}^2)$ rectangles.
Therefore, there are $\mathcal{O}(\mathcal{R}^2)$ points in $4$-dimension. Since, $\mathcal R = \mathcal O(n^2)$ in the worst case, sorting the points require $\mathcal{O}(\mathcal{R}^2\log{\mathcal{R}^2}) = \mathcal{O}(\mathcal{R}^2\log{n})$ time. Since the coordinate values are bounded within the range $1$ to $n$, we can sort them using counting sort algorithm. So this will reduce the sorting time to $\mathcal{O}(\mathcal{R}^2)$. Constructing a $3$-dimensional multi-level balanced binary search tree from $\mathcal{O}(\mathcal{R}^2)$ points takes $\mathcal{O}(\mathcal{R}^2\log^3{\mathcal{R}^2})$ = $\mathcal{O}(\mathcal R^2\log^3{n})$ time. Each query in tree of can be done in $\mathcal{O}(\log^3{n})$ time. Now, for $\mathcal{O}(\mathcal{R}^2)$ points, a total of $\mathcal{O}(\mathcal{R}^2)$ queries are made which takes a total of $\mathcal{O}(R^2\log^3{n})$ time. Therefore, the overall running time of our algorithm is  $\mathcal{O}(\mathcal{R}^2\log^3{n})$. \qed
\end{proof}

Since $\mathcal{R} = \mathcal{O}(n^2)$, the running time of our algorithm becomes $\mathcal{O}(n^4\log^3{n})$ in the worst case, which is not better than that of the Dynamic Programming algorithm ($\mathcal{O}(n^4)$). But in cases where we have $\mathcal{R} = \mathcal{O}(n)$ it exhibits very good performance. In such case the running time reduces to $\mathcal{O}(n^2\log^3{n})$. Even for $\mathcal R = \mathcal{O}(n^{1.5})$ this algorithm performs better ($\mathcal{O}(n^3\log^3{n})$) than the DP algorithm.

\section{Conclusion and Future Works}
\label{fwc}
In this paper, we have introduced and studied the longest common palindromic subsequence (LCPS) problem, which is a variant of the classic LCS problem. We have first presented a dynamic programming algorithm to solve it, which runs in $\mathcal{O}(n^4)$ time. Then, we have identified and studied some interesting relation of the problem with computational geometry and devised an $\mathcal{O}(\mathcal{R}^2\log^3{n})$ time algorithm. In our results, we have assumed that the two input strings are of equal length $n$. However, our results can be easily extended for the case where the two input strings are of different lengths. To the best of our knowledge this is the first attempt in the literature to solve this problem. 
Further research can also be carried out towards studying different other variants of the LCPS problem.


\bibliographystyle{splncs03}
\bibliography{paper}

\begin{thebibliography}{10}
\providecommand{\url}[1]{\texttt{#1}}
\providecommand{\urlprefix}{URL }

\bibitem{Apo95}
Apostolico, A., Breslauer, D., Galil, Z.: Parallel detection of all palindromes
  in a string. Theoretical Computer Science  141,  163 -- 173 (April 1995)

\bibitem{Bre95}
Breslauer, D., Galil, Z.: Finding all periods and initial palindromes of a
  string in parallel. Algorithmica  14,  355 -- 366 (October 1995)

\bibitem{t}
Chen, K.Y., Hsu, P.H., Chao, K.M.: Identifying approximate palindromes in
  run-length encoded strings. In: Proceedings of 21st International Symposium,
  ISAAC 2010, Jeju, Korea, December 15-17, 2010. pp. 339 -- 350 (2010)

\bibitem{h}
Eltayeb, F., Elbahir, M., Mohamed, S., Ahmed, M., Zaki, N.: Development of a
  web-base application to detect palindromes in dna sequence. In: Proceedings
  of 4th International Conference on Innovations in Information Technology. pp.
  725 -- 727 (2007)

\bibitem{Galil76}
Galil, Z.: Real-time algorithms for string-matching and palindrome recognition.
  In: Proceedings of the eighth annual ACM symposium on Theory of computing.
  pp. 161--173 (1976)

\bibitem{l}
Gusfield, D.: Algorithms on Strings, Trees, and Sequences: Computer Science and
  Computational Biology. Cambridge University Press, New York

\bibitem{j}
Hoffmann, M., Rychlewski, J.: Searching for palindromic sequences in primary
  structure of proteins. Computational Methods in Science and Technology pp. 21
  -- 24 (1999)

\bibitem{p}
Hsu, P.H., Chen, K.Y., Chao, K.M.: Finding all approximate gapped palindromes.
  In: Proceedings of 20th International Symposium, ISAAC 2009, Honolulu,
  Hawaii, USA, December 16-18, 2009. pp. 1084 -- 1093 (2009)

\bibitem{o}
Kolpakova, R., Kucherov, G.: Searching for gapped palindromes. Theoretical
  Computer Science pp. 5365 -- 5373 (November 2009)

\bibitem{k}
Manacher, G.: A new linear-time "on-line" algorithm for finding the smallest
  initial palindrome of a string. Journal of the ACM  22,  346 -- 351 (July
  1975)

\bibitem{Mart09}
Martinek, T., Vozenilek, J., Lexa, M.: Architecture model for approximate
  palindrome detection. Design and Diagnostics of Electronic Circuits and
  Systems  0,  90--95 (2009)

\bibitem{i}
Martínek, T., Lexa, M.: Hardware acceleration of approximate palindromes
  searching. In: Proceedings of The International Conference on
  Field-Programmable Technology. pp. 65 -- 72 (2008)

\bibitem{n}
Matsubara, W., Inenaga, S., Ishino, A., Shinohara, A., Nakamura, T., Hashimoto,
  K.: Efficient algorithms to compute compressed longest common substrings and
  compressed palindromes. Theoretical Computer Science  410,  900--913 (March
  2009)

\bibitem{m}
Porto, A.H.L., Barbosa, V.C.: Finding approximate palindromes in strings.
  Pattern Recognition  (2002)

\bibitem{u}
Tomohiro, I., Shunsuke, I., Masayuki, T.: Palindrome pattern matching. In:
  Proceedings of 22nd Annual Symposium, CPM 2011, Palermo, Italy, June 27-29,
  2011. pp. 232 -- 245 (2011)

\end{thebibliography}

\end{document}